\documentclass[preprint,aps,onecolumn,superscriptaddress,floatfix,
nofootinbib,showpacs,longbibliography]{revtex4-1}

\usepackage[utf8]{inputenc}  
\usepackage[T1]{fontenc}     
\usepackage[british]{babel}  
\usepackage[sc,osf]{mathpazo}
\usepackage[scaled=0.86]{berasans}  
\usepackage[colorlinks=true, citecolor=blue, urlcolor=blue]{hyperref}  
\usepackage{scrextend}
\usepackage{graphicx} 
\usepackage[babel]{microtype}  
\usepackage{amsmath,amssymb,amsthm,bm,amsfonts,mathrsfs,bbm} 

\usepackage{xspace}  
\usepackage{pgf,tikz}
\usepackage{xcolor,colortbl}
\usepackage{array}
\usepackage{bigstrut}
\usepackage{braket}
\usepackage{color}
\usepackage{natbib}
\usepackage{multirow}

\newcommand{\be}{\begin{equation}}
\newcommand{\ee}{\end{equation}}
\newcommand{\ba}{\begin{eqnarray}}
\newcommand{\ea}{\end{eqnarray}}

\newtheorem{theorem}{Theorem}
\newtheorem{corollary}{Corollary}
\newtheorem{definition}{Definition}

\newtheorem{observation}{Observation}

\newtheorem{remark}{Remark}
\newtheorem{lemma}{Lemma}

\begin{document}

\title{Family of bound entangled states on the boundary of Peres set} 

\author{Saronath Halder}
\email{saronath.halder@gmail.com}    
\affiliation{Department of Mathematics, Indian Institute of Science Education and Research Berhampur, Transit Campus, Government ITI, Berhampur 760010, India.} 
  
\author{Manik Banik}
\email{manik11ju@gmail.com}    
\affiliation{S.N. Bose National Center for Basic Sciences, Block JD, Sector III, Salt Lake, Kolkata 700098, India.}

\author{Sibasish Ghosh}
\email{sibasish@imsc.res.in} 
\affiliation{Optics \& Quantum Information Group, The Institute of Mathematical Sciences, HBNI, CIT Campus,
Taramani,  Chennai  600113, India}

\begin{abstract}
Bound entangled (BE) states are strange in nature: non-zero amount of free entanglement is required to create them but no free entanglement can be distilled from them under local operations and classical communication (LOCC). Even though  usefulness of such states has been shown in several information processing tasks, there exists no simple method to characterize them for an arbitrary composite quantum system.
Here we present a $(d-3)/2$-parameter family of BE states each with positive partial transpose (PPT). 
This family of PPT-BE states is introduced by constructing an unextendible product basis (UPB) in $\mathbb{C}^d\otimes\mathbb{C}^d$ with $d$ odd and $d\ge 5$. The range of each such PPT-BE state is contained in a $2(d-1)$ dimensional entangled subspace whereas the associated UPB-subspace is of dimension $(d-1)^2+1$. 
We further show that each of these PPT-BE states can be written as a convex combination of $(d-1)/2$ number of rank-4 PPT-BE states. Moreover, we prove that these rank-4 PPT-BE states are extreme points of the convex compact set $\mathcal{P}$ of all PPT states in $\mathbb{C}^d\otimes\mathbb{C}^d$, namely the {\it Peres} set. An interesting geometric implication of our result is that the convex hull of these rank-4 PPT-BE extreme points -- the  $(d-3)/2$-simplex -- is sitting on the boundary between the set $\mathcal{P}$ and the set of non-PPT states. We also discuss consequences of our construction in the context of quantum state discrimination by LOCC. 
\end{abstract}



\maketitle

\section{Introduction}
\label{sec:introduction}
Entanglement is one of the fundamental features of multipartite quantum systems. Though this very concept was recognized in the early days of quantum theory \cite{schro}, its physical meaning remains elusive  till date. Advent of quantum information theory identifies quantum entanglement as a useful resource for several information processing tasks [see \cite{horo1,gune} and references therein].
Thus characterization, detection, and quantification of quantum entanglement are of practical relevance and it is one of the main objectives to pursue research in quantum information theory. 
\par
It has been shown that the {\it quantum separability problem}, i.e., to verify whether an arbitrary density matrix of a given bipartite quantum system is entangled or separable is an NP-hard problem \cite{gurvits}. But we have some sufficient criteria to detect entanglement of a given state. One such useful criterion is negative partial transposition (NPT): given a bipartite quantum state if its partial transposition is non-positive then the state is entangled \cite{peres}. However, positive partial transposition (PPT) does not always guarantee separability, indeed there exist PPT-entangled states \cite{horo2}.
\par
Quantifying entanglement is another challenging aspect in entanglement theory and different operational as well as geometric measures have been introduced so far for this purpose. One such operationally motivated measure is {\it entanglement of distillation} \cite{ben2,ben1,ben3}. It is defined as the optimal rate of obtaining pure entangled state (singlet state) given many (asymptotically large) copies of the noisy (impure) entangled state under a sequence of local operations and classical communication (LOCC). However, this measure is not {\it faithful}\footnote{A {\it faithful} entanglement measure is strictly positive if and only if the state is entangled.} as it has been shown that PPT entangled states are undistillable and hence are also called bound entangled (BE) states \cite{horo2}. These PPT-BE states may not be the only type of BE states as it has been conjectured that there may exist NPT-BE states \cite{divin1,Dur}.
\par
During last few years usefulness of PPT-BE states has been extensively studied in different contexts, e.g., entanglement activation \cite{Horodecki(1),Vollbrecht}, probabilistic inter convertibility
among multipartite pure states \cite{Ishizaka-1,Ishizaka-2}, universal usefulness \cite{Masanes}, secure key distillation \cite{Horodecki(2),Horodecki(3),Horodecki(4)}, quantum metrology \cite{Czekaj}. Their connections to quantum steering \cite{Moroder} as well as to quantum nonlocality \cite{Vertesi} have also been established. However, due to the hardness of {\it quantum separability problem} there is no simple scheme to decide whether a given state is PPT-BE or not. We only know some examples and special constructions of such states \cite{horo2,be-2,be-3,be-4,be-5,be-6,be-7,ben4,divin2,Yu17,Sentis18}. One elegant construction comes from the structure of unextendible product bases (UPBs). It was shown that the normalized projector onto the subspace orthogonal to the UPB-subspace (subspace spanned by a UPB) is a PPT-BE state \cite{ben4,divin2}. 
\par
In order to explore the geometry of the set of PPT-BE states, researchers have considered following sets of density matrices: (i) the set $\mathcal{S}$ of all separable density matrices, (ii) the set $\mathcal{P}$ of all PPT density matrices. Both the sets are convex and compact. These sets are identical for any two-qubit system as well as for any qubit-qutrit system while for other systems  $\mathcal{P}$ strictly contains $\mathcal{S}$. Clearly, $\mathcal{S}$ contains only the rank-1 separable density matrices (pure product states) as its extreme points. However, except the two-qubit and the qubit-qutrit systems $\mathcal{P}$ contains not only the extreme points of $\mathcal{S}$ but also some additional extreme points. These new extreme points are nothing but PPT-BE states, and identifying such PPT-BE states is a troublesome task. Moreover, there exits some PPT-BE states, called {\it edge states}, living on the boundary of the set $\mathcal{P}$ and that of NPT states. To understand the extremely complicated structure of $\mathcal{P}$ it is important to identify the extremal PPT-BE states as well as the {\it edge states}. In last few years a considerable effort has been given addressing this question. A series of interesting results can be found in \cite{lewen, Szarek, leinaas1, leinaas2, leinaas3, augusiak, Sollid, chen, hansen, garberg, Chen1, Ha}. In \cite{lewen}, entanglement witness operators have been constructed for edge states. In \cite{Szarek}, it has been shown that for any bipartite system the ratio between the probabilities of finding a PPT state in the interior of the set of PPT mixed quantum states and at its boundary is equal to 2. Later in \cite{leinaas1, leinaas2, leinaas3, augusiak, Sollid, chen, hansen, garberg, Chen1, Ha} different methods have been proposed to identify the extremal PPT-BE states of the set $\mathcal{P}$. In particular, Chen and \DJ{}okovi\'{c} have shown that all rank-4 two-qutrit PPT-BE states can be constructed from unextendible product bases and all such PPT-BE states are extreme points of $\mathcal{P}$ \cite{chen}. For a higher dimensional system, the problem of identifying the extremal PPT-BE states as well as the edge states becomes more complicated, indeed very little is known so far. The main goal of the present work is to understand the set $\mathcal{P}$ for higher dimensional systems by exploring new classes of extreme points and edge states. In the following, we summarize the main findings of this paper: 
\begin{itemize}
\item We construct new UPB in $\mathbb{C}^d\otimes\mathbb{C}^d$, where $d\ge5$ and $d$ is odd. We also provide the tile structures corresponding to the UPB. The cardinality of such a UPB is $(d-1)^2+1$ and the corresponding entangled subspace is of dimension $2(d-1)$. The present construction is a generalization of the Tiles UPB in $\mathbb{C}^3\otimes\mathbb{C}^3$, given in \cite{ben4}.
\item We show that the PPT-BE state proportional to the full rank projector onto the entangled subspace allows a convex decomposition in terms of $(d-1)/2$ number of rank-4 PPT-BE states and hence is not an extreme point of the set $\mathcal{P}$. Interestingly, it turns out that the rank-4 PPT-BE states appeared in the above decomposition are extreme points of the set $\mathcal{P}$.
\item We further show that any convex mixture of the aforesaid extreme points are {\it edge states}. Geometrically, a $(d-3)/2$ simplex is formed by the aforesaid $(d-1)/2$ number of rank-4 extreme points and the simplex resides on the boundary between the set $\mathcal{P}$ and the set of NPT states. At this point the result of \cite{chen} is worthy to mention. It turns out that the entangled subspace corresponding to a two-qutrit UPB contains only one edge state and hence it is also an extreme point of $\mathcal{P}$.   

\item We study the cardinality of different locally indistinguishable sets (both {\it completable} and {\it uncompletable}) of orthogonal product states. We also discuss the merits of our construction in the context of orthogonal mixed state discrimination by LOCC.  
\end{itemize}

The manuscript is organized in the following way. In Section \ref{notation} we provide the notations used here and discuss about the prerequisite ideas. In Section \ref{results} we provide the main results where we first briefly review the Tiles UPB in $\mathbb{C}^3\otimes\mathbb{C}^3$ and then present the generalized Tiles UPB in $\mathbb{C}^5\otimes\mathbb{C}^5$ and $\mathbb{C}^d\otimes\mathbb{C}^d$ respectively. Then we provide the parametric family of PPT-BE edge states. Finally, in Section \ref{discussion} we make concluding remarks with some open problems for further research.

\section{Notations and preliminaries}\label{notation}
A complex Hilbert space of dimension $d$ is denoted by $\mathbb{C}^d$. To represent a (unnormalized) non-null vector, we use the {\it ket} notation $|v\rangle\in\mathbb{C}^d$, while $|\tilde{v}\rangle$ denotes the normalized vector parallel to $|v\rangle$. A linear operator maps a vector $|v\rangle\in\mathbb{C}^d$ to another vector $|v^\prime\rangle\in\mathbb{C}^d$. Rank of a linear operator $\mathcal{O}: \mathbb{C}^d\mapsto\mathbb{C}^d$ is dimension of its range denoted by $\mathcal{R}(\mathcal{O})$. Given a set of vectors, $S=\{|v_1\rangle,...,|v_k\rangle\}\subset\mathbb{C}^d$, their linear span form a subspace $Span(S):=\left\lbrace \sum_{i=1}^{k}\alpha_i|v_i\rangle~|~\alpha_i\in\mathbb{C},~\forall~i\right\rbrace$; sometime we will use the notation $|\sum_{i=1}^{k}\alpha_iv_i\rangle\equiv\sum_{i=1}^{k}\alpha_i|v_i\rangle$.   

Tensor product of two Hilbert spaces $\mathbb{C}^{d_A}$ and $\mathbb{C}^{d_B}$ is denoted by $\mathbb{C}^{d_A}\otimes\mathbb{C}^{d_B}$. Consider an $n$-dimensional subspace of $\mathbb{C}^{d_A}$ spanned by a set of orthogonal vectors $\{|u_i\rangle\}_{i=1}^{n}$ and an $m$-dimensional subspace of $\mathbb{C}^{d_B}$ spanned by $\{|v_i\rangle\}_{i=1}^{m}$. We say that $\{|u_1\rangle_A,...,|u_{n}\rangle_A\}\otimes\{|v_1\rangle_B,...,|v_{m}\rangle_B\}$ spans the subspace $\mathbb{C}^{n}\otimes\mathbb{C}^{m}$ of $\mathbb{C}^{d_A}\otimes\mathbb{C}^{d_B}$, where $\{|u_1\rangle_A,...,|u_{n}\rangle_A\}\otimes\{|v_1\rangle_B,...,|v_{m}\rangle_B\}\equiv\{|u_i\rangle_A\otimes|v_j\rangle_B\}_{i,j=1}^{n,m}$. 

A convex set $\mathcal{A}$ is a subset of an affine space that is closed under convex combinations, i.e., for any $a_i\in\mathcal{A}$,$\sum_{i=1}^{n}p_ia_i\in\mathcal{A}$, where $p_i\ge 0,~\forall~i$ and $\sum_{i=1}^{n}p_i=1$. A point $b\in\mathcal{A}$ is called an extreme point of $\mathcal{A}$ if it cannot be expressed as convex combination of other points in $\mathcal{A}$. The set of all extreme points of $\mathcal{A}$ is denoted by $\mathcal{E(A)}$. A subset in Euclidean space is called compact if it is closed (contains all limit points) and bounded. According to Krein-Milman theorem \cite{Krein}, any convex compact set of a finite dimensional vector space is equal to the convex hull of its extreme points. Thus, this theorem assures that while maximizing a linear functional over a convex compact set, it is sufficient to scan over only the extreme points instead of the whole convex compact set.

Every quantum system is associated with a Hilbert space. The state of a $d$-level quantum system is described by a density matrix $\rho$ which is a positive semidefinite, hermitian, trace-1 operator acting on $\mathbb{C}^d$. Set of all these density matrices $\mathcal{D}(\mathbb{C}^d)$ forms a convex compact subset of a real Euclidean space $\mathbb{R}^{(d^2-1)}$. The density matrices of rank-1, i.e., $\rho=|\psi\rangle\langle\psi|$, where $|\psi\rangle\in\mathbb{C}^d$ constitutes $\mathcal{E(D)}$. 

A bipartite quantum system is associated with a tensor product Hilbert space  $\mathcal{H}=\mathcal{H}_A\otimes\mathcal{H}_B=\mathbb{C}^{d_A}\otimes\mathbb{C}^{d_B}$; where $d_k$ is the dimension of $\mathcal{H}_k$, $k\in\{A,B\}$. A quantum state $\rho_{AB}\in\mathcal{D}(\mathbb{C}^{d_A}\otimes\mathbb{C}^{d_B})$ is called a separable state (or product state) if it can be written as, $\rho_{AB}=\sum_ip_i\rho_A^i\otimes\rho_B^i$, where $\rho_k^i\in\mathcal{D}(\mathbb{C}^{d_k}),~\forall~i,k$; $p_i\in\{0,1\}$ and $\sum_ip_i=1$. States that cannot be expressed in this form are entangled. The set of all separable states $\mathcal{S}(\mathbb{C}^{d_A}\otimes\mathbb{C}^{d_B})$ is also a convex compact set which is strictly contained in $\mathcal{D}$, i.e., $\mathcal{S}\subset\mathcal{D}$. Note that $\mathcal{E(D)}$ is constituted by both pure product and entangled states while only the pure product states constitute $\mathcal{E(S)}$.

To certify entanglement of a bipartite state, partial transpose operation plays a crucial role. Partial transpose of a density matrix 
$\rho_{AB}$ is denoted by $\rho_{AB}^{T_k}$; where $T_k$ is transposition operation in a chosen basis with respect to $k^{th}$ party.
If $\rho_{AB}^{T_k}\ngeq 0$, $\rho_{AB}$ must be entangled \cite{peres}. However, $\rho_{AB}^{T_k}\ge 0$ guarantees separability of a given density matrix for the systems $\mathbb{C}^2\otimes\mathbb{C}^2$, $\mathbb{C}^3\otimes\mathbb{C}^2$, and $\mathbb{C}^2\otimes\mathbb{C}^3$. In fact, in higher dimensions, there exist entangled states with positive partial transpose \cite{horo2}. For a composite Hilbert space, the states having positive partial transpose again form a convex compact set $\mathcal{P}(\mathbb{C}^{d_A}\otimes\mathbb{C}^{d_B})$ (say), also known as the {\it Peres set} \cite{leinaas1}. Clearly, $\mathcal{S}\subseteq\mathcal{P}\subset\mathcal{D}$, with set equality holds true for lower dimensions, i.e., for any two-qubit and qubit-qutrit systems. Thus, $\mathcal{E(P)}$ is exactly same as $\mathcal{E(S)}$ in these dimensions and $\mathcal{E(P)}$ is strictly bigger than $\mathcal{E(S)}$ for all other dimensions. So, the nontriviality of characterizing $\mathcal{E(P)}$ lies in the fact that it contains not only pure product states but also some additional PPT-BE states. One of the goals of the present work is to understand these additional PPT-BE states of $\mathcal{E(P)}$. Here, these PPT-BE states are connected to UPBs, definition of which is given below.

\begin{definition}
Consider a bipartite quantum system $\mathcal{H}=\mathbb{C}^{d_A}\otimes\mathbb{C}^{d_B}$. A complete orthogonal product basis (COPB) is a set of orthogonal product states that spans $\mathcal{H}$ while an incomplete orthogonal product basis (ICOPB) is a set of pure orthogonal product states that spans a subspace $\mathcal{H}_S$ of $\mathcal{H}$. An uncompletable product basis (UCPB) is an ICOPB whose complementary subspace $\mathcal{H}_S^\perp$ contains fewer pairwise orthogonal pure product states than its dimension. An unextendible product basis (UPB) is a UCPB, whose complementary subspace $\mathcal{H}_S^\perp$ contains no product states. 
\end{definition}

Note that $\mathcal{H}_S^\perp$, in case of a UPB, is a fully entangled subspace. Here, we denote such a subspace as $\mathcal{H}_E$. Bennett {et al.} showed that the normalized projector onto the subspace $\mathcal{H}_E$ for a given UPB is a PPT-BE state \cite{ben4}. These PPT-BE states are known to be {\it edge states} as they reside on the boundary between the set of PPT-BE states and that of the NPT states. In the following, we recall the mathematical definition of an {\it edge state} \cite{lewen}.

\begin{definition}\label{definition-2}
A PPT-BE state $\delta_{AB}$ is an edge state {\it iff} there exists no product state $|\varphi_A\rangle|\varphi_B\rangle$ and $\epsilon>0$, such that $\delta_{AB}-\epsilon\mathbb{P}(|\varphi_A\rangle|\varphi_B\rangle)$ is positive or does have a PPT, where $\mathbb{P}(\cdot)$ denotes the one dimensional  projection operator.  
\end{definition}

A necessary and sufficient criterion for an edge state is given in \cite{lewen}. We recall that criterion in the following remark.

\begin{remark}\label{remark-1}
A PPT-BE state $\delta_{AB}$ is an edge state {\it iff} there exists no $|\varphi_A\rangle|\varphi_B\rangle\in\mathcal{R}(\delta_{AB})$, such that $|\varphi_A\rangle|\varphi_B^\star\rangle\in\mathcal{R}(\delta_{AB}^{T_B})$, where $|\varphi_B^\star\rangle$ is the complex-conjugate of $|\varphi_B\rangle$.
\end{remark}

Obviously, any PPT-BE state belonging to $\mathcal{E(P)}$ is an {\it edge state}, but the converse is not true in general. Still, identifying the {\it edge states} is important to decipher the complicated geometrical structure of the set $\mathcal{P}$. This is another aspect of the present study. In particular, we provide parametric family of edge states and the number of parameters increase linearly with the dimension of subsystems. 

Another important aspect of UPBs is that they exhibit the phenomenon {\it quantum nonlocality without entanglement} \cite{ben99} and hence the orthogonal pure product states within a UPB cannot be perfectly distinguished by LOCC \cite{ben4, divin2}. Suppose, no party can perform {\it nontrivial} and {\it orthogonality preserving} measurement\footnote{If not all the positive operator valued measure (POVM) elements describing a measurement are proportional to identity operator then the measurement is a nontrivial measurement. Moreover, while distinguishing a given set of orthogonal states if the post measurement states remain pairwise orthogonal then it is an orthogonality preserving measurement.} \cite{Groisman, wal02, ye07} in order to distinguish a set of orthogonal pure product states. Then this guarantees that the states of the given set cannot be distinguished perfectly by LOCC. Indeed, not even a single state from that set can be perfectly identified by such measurements. Again, in the context of orthogonal mixed state discrimination, UPBs play a crucial role. It has been shown that any state supported in $\mathcal{H}_E$ cannot be {\it conclusively} distinguished from the mixed state proportional to the projector onto the UPB-subspace \cite{ban1}. Our construction leads to a few interesting observations regarding state discrimination problem by LOCC.    

\section{Results}\label{results}
Finding all the extreme points of the set $\mathcal{P}$ of a given system $\mathbb{C}^{d_A}\otimes\mathbb{C}^{d_B}$ is a highly nontrivial task. This is because of two reasons: firstly, identifying a PPT-BE states is itself a nontrivial job, and then determining whether such a PPT-BE state is an extreme point of the set $\mathcal{P}$ is the second hurdle. Leinaas and co-authors, for the first time,  derived a necessary and sufficient condition for uniquely identifying the extreme points of $\mathcal{P}$ \cite{leinaas1}. However, for useful implication of their condition it requires an algorithmic search to detect any such point. Subsequently, this method has been studied in several bipartite and multipartite systems \cite{leinaas2,leinaas3}. Later, these works motivated Chen and \DJ{}okovi\'{c} to come up with an analytical approach to explore the nontrivial extreme points of $\mathcal{P}$ \cite{chen}. In particular, using the techniques of projective geometry, they proved that any rank-4 PPT-BE state in $\mathbb{C}^3\otimes\mathbb{C}^3$ is an extreme point of $\mathcal{P}$. As an immediate extension of this result, we give the following lemma.

\begin{lemma}\label{lemma1}
	Consider a rank-4 PPT-BE state $\rho_{AB}$ of $\mathbb{C}^{d_A}\otimes\mathbb{C}^{d_B}$ with $d_A,d_B\ge 3$. Assume that the range of $\rho_{AB}$ is supported in $\mathcal{H}^\prime_A\otimes\mathcal{H}^\prime_B$, where $\mathcal{H}^\prime_A~ (\mathcal{H}^\prime_B)$ is a three dimensional subspace of $\mathbb{C}^{d_A}~(\mathbb{C}^{d_B})$ and in $\mathcal{H}^\prime_A\otimes\mathcal{H}^\prime_B$, the tensor product is the induced version of `$\otimes$' used for the full Hilbert space $\mathbb{C}^{d_A}\otimes\mathbb{C}^{d_B}$. Then the state $\rho_{AB}$ is an extreme point of the set $\mathcal{P}$ of $\mathbb{C}^{d_A}\otimes\mathbb{C}^{d_B}$.
\end{lemma}

\begin{proof}
In contradiction to the statement of the above lemma, let us assume that $\rho_{AB}$ is not an extreme point of $\mathcal{P}$. Therefore, $\rho_{AB}$ allows at least one decomposition of the form $\rho_{AB}=p\eta_{AB}+(1-p)\eta_{AB}^\prime$, where $p\in(0,1)$ and $\rho_{AB}\ne\eta_{AB}$, $\rho_{AB}\ne\eta_{AB}^\prime$ and $\eta_{AB}\ne\eta_{AB}^\prime$. The ranges of $\eta_{AB}$ and $\eta_{AB}^\prime$ are fully contained in the range of $\rho_{AB}$. Consider the projector $\mathbb{P}^\prime$ onto the subspace $\mathcal{H}^\prime_A\otimes\mathcal{H}^\prime_B$. Since $\mathbb{P}^\prime\sigma\mathbb{P}^\prime=\sigma$, for $\sigma\in\{\rho_{AB},\eta_{AB},\eta_{AB}^\prime\}$, therefore $\mathbb{P}^\prime\rho_{AB}\mathbb{P}^\prime=p\mathbb{P}^\prime\eta_{AB}\mathbb{P}^\prime+(1-p)\mathbb{P}^\prime\eta_{AB}^\prime\mathbb{P}^\prime$, which contradicts the result of Chen et al. \cite{chen} that any PPT-BE state of rank-4 in $\mathbb{C}^3\otimes\mathbb{C}^3$ is an extrema point. 
\end{proof}

Next we construct a $(d-3)/2$-parameter family of PPT-BE states in $\mathbb{C}^{d}\otimes\mathbb{C}^{d}$ for odd $d$ and corresponding UPBs. Our construction is a generalization of Tiles UPB in $\mathbb{C}^{3}\otimes\mathbb{C}^{3}$ introduced by Bennett et al \cite{ben1}. So, before presenting our main results we first briefly review different aspects of the Tiles UPB in $\mathbb{C}^{3}\otimes\mathbb{C}^{3}$.
\subsection{Tiles UPB in $\mathbb{C}^{3}\otimes\mathbb{C}^{3}$}\label{sec3x3}
The orthogonal pure product states forming the Tiles UPB in $\mathbb{C}^{3}\otimes\mathbb{C}^{3}$ are given below:   
\begin{equation}\label{UPB-3}
\begin{array}{c}
|\psi_{1}\rangle=|0\rangle|0 - 1\rangle,~~
|\psi_{2}\rangle=|2\rangle|1 - 2\rangle,\\[1ex]

|\psi_{3}\rangle=|0 - 1\rangle|2\rangle,~~
|\psi_{4}\rangle=|1 - 2\rangle|0\rangle,\\[1ex]

|S\rangle=|0+1+2\rangle|0+1+2\rangle.\\[1ex]
\end{array}
\end{equation}

\begin{figure}[ht]
	\includegraphics{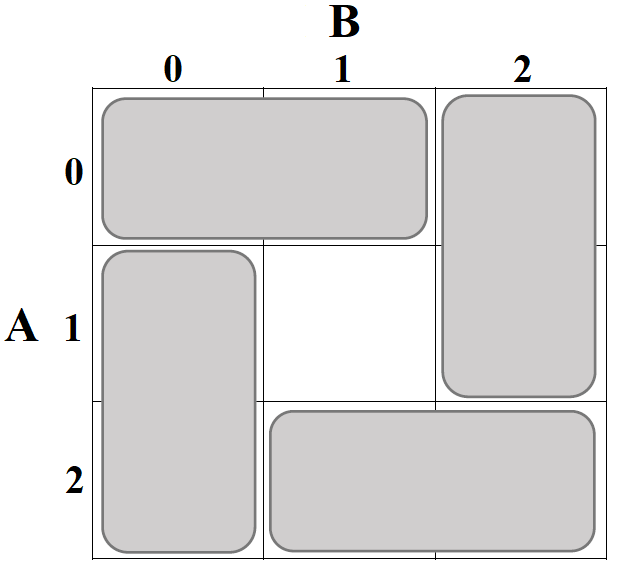}
	\caption{Tile structure of a UPB in $\mathbb{C}^3\otimes\mathbb{C}^3$.}\label{3x3}
\end{figure}

The states are written without the normalization coefficients and through out the paper we follow this convention unless stated otherwise. We say a subspace spanned by the product states forming a UPB is a UPB-subspace ($\mathcal{H}_U$). The subspace orthogonal to $\mathcal{H}_U$ contains no product state and hence is an entangled-subspace ($\mathcal{H}_E$). The normalized projector on $\mathcal{H}_E$ is given by-
\begin{equation}\label{bound-3}
\rho_3=\frac{1}{4}\left(\mathbb{I}_3\otimes\mathbb{I}_3-\sum_{i=1}^4|\tilde{\psi}_i\rangle\langle\tilde{\psi}_i|-|\tilde{S}\rangle\langle \tilde{S}|\right),
\end{equation}
where $\mathbb{I}_n$ denotes $n\times n$ identity matrix, $|\tilde{\psi}_i\rangle$'s and $|\tilde{S}\rangle$ are normalized states of that given in Eq.(\ref{UPB-3}). We use the subscript of $\rho_3$ to indicate that it is a density matrix on $\mathbb{C}^{3}\otimes\mathbb{C}^{3}$. Since, $\mathcal{H}_E$ contains no product state therefore $\rho_3$ is an entangled state. To prove positivity of $\rho_3$ under partial transpose we make use of the following observation taken from \cite{divin1}. 

\begin{observation}\label{obs1}
	Under transposition on Alice's side any pure product state $|\alpha_A\rangle\langle\alpha_A|\otimes|\alpha_B\rangle\langle\alpha_B|$ becomes $|\alpha^\ast_A\rangle\langle\alpha^\ast_A|\otimes|\alpha_B\rangle\langle\alpha_B|$ and hence a set of orthogonal product states are mapped into another set of orthogonal product states. 
\end{observation}

This guarantees the positivity of $\rho_3$ under partial transposition. If we remove the state $|S\rangle$ from Eq.(\ref{UPB-3}), the remaining four states together with the following five product states form a complete orthogonal product basis (COPB): 
\begin{equation}\label{UPB-3-C-onb}
\begin{array}{c}
|\psi_{5}\rangle=|0\rangle|0 + 1\rangle,~~
|\psi_{6}\rangle=|2\rangle|1 + 2\rangle,\\[1ex]

|\psi_{7}\rangle=|0 + 1\rangle|2\rangle,~~
|\psi_{8}\rangle=|1 + 2\rangle|0\rangle,\\[1ex]

|\psi_{9}\rangle=|1\rangle|1\rangle.\\[1ex]
\end{array}
\end{equation}

When the state $|S\rangle$ is considered with the set of states $\{|\psi_i\}_{i=1}^4$, it stops inclusion of any the above states to form a COPB. In other words it plays the role of `stopper' to construct the UPB. Using the states of Eq.(\ref{UPB-3-C-onb}) one can construct four pairwise orthogonal entangled states $\{|\phi_i\rangle\}_{i=1}^4$ that span $\mathcal{H}_E$. One such construction is given as follows-   
\begin{equation}\label{UPB-3-ent}
\begin{array}{l}
|\phi_{1}\rangle =|\psi_{5}\rangle+|\psi_{6}\rangle-|\psi_{7}\rangle-|\psi_{8}\rangle,\\[1ex]
|\phi_{2}\rangle =|\psi_{5}\rangle-|\psi_{6}\rangle+|\psi_{7}\rangle-|\psi_{8}\rangle,\\[1ex]
|\phi_{3}\rangle =|\psi_{5}\rangle-|\psi_{6}\rangle-|\psi_{7}\rangle+|\psi_{8}\rangle,\\[1ex]
|\phi_{4}\rangle =a_1(|\psi_{5}\rangle+|\psi_{6}\rangle+|\psi_{7}\rangle+|\psi_{8}\rangle)+a_2|\psi_{9}\rangle.
\end{array}
\end{equation}
Note that the orthogonality of the states $\{|\psi_{i}\rangle\}_{i=1}^4$ with that of Eq.(\ref{UPB-3-ent}) is immediate from construction. Orthogonality of 
$\{|\phi_{i}\rangle\}_{i=1}^3$ with $|S\rangle$ is also assured due to the construction but to make $|\phi_{4}\rangle$ orthogonal to $|S\rangle$ we need to fix the values of $a_1, a_2$ accordingly.  Now the PPT-BE state of Eq.(\ref{bound-3}) can be rewritten as $\rho_3=1/4\sum_{i=1}^4|\tilde{\phi}_i\rangle\langle\tilde{\phi}_i|$. With reference to this context, it is quite worthy to mention that $\rho^\prime_3=\sum_{i=1}^4p_i|\tilde{\phi}_i\rangle\langle\tilde{\phi}_i|$ is NPT, where $p_i$'s are probabilities (excluding the case when all $p_i$'s are equal) \cite{chen}. Indeed all such states are one-copy distillable \cite{chen1}. Next we generalize this tile structure to higher dimensional Hilbert spaces and explore different intriguing aspects of such generalization.

\subsection{Tiles UPB in $\mathbb{C}^{5}\otimes\mathbb{C}^{5}$}\label{sec5x5}
In $\mathbb{C}^{5}\otimes\mathbb{C}^{5}$, it is possible to construct a COPB based on the title structure given in Fig.\ref{5x5}. If we choose a suitable stopper $|S\rangle\in\mathbb{C}^{5}\otimes\mathbb{C}^{5}$ and remove the product states that are not orthogonal to $|S\rangle$ then the remaining states of the COPB along with $|S\rangle$ form a UPB in $\mathbb{C}^{5}\otimes\mathbb{C}^{5}$. Such a UPB is given by-

\begin{equation}\label{UPB-5}
\begin{array}{l}
|\psi_{1}\rangle=|0\rangle|0 - 1 + 2 - 3\rangle,~~~
|\psi_{2}\rangle=|0\rangle|0 + 1 - 2 - 3\rangle,\\[1ex]
|\psi_{3}\rangle=|0\rangle|0 - 1 - 2 + 3\rangle,~~~

|\psi_{4}\rangle=|4\rangle|1 - 2 + 3 - 4\rangle,\\[1ex]
|\psi_{5}\rangle=|4\rangle|1 + 2 - 3 - 4\rangle,~~~
|\psi_{6}\rangle=|4\rangle|1 - 2 - 3 + 4\rangle,\\[1ex]

|\psi_{7}\rangle=|0 - 1 + 2 - 3\rangle|4\rangle,~~~
|\psi_{8}\rangle=|0 + 1 - 2 - 3\rangle|4\rangle,\\[1ex]
|\psi_{9}\rangle=|0 - 1 - 2 + 3\rangle|4\rangle,~~~

|\psi_{10}\rangle=|1 - 2 + 3 - 4\rangle|0\rangle,\\[1ex]
|\psi_{11}\rangle=|1 + 2 - 3 - 4\rangle|0\rangle,~~
|\psi_{12}\rangle=|1 - 2 - 3 + 4\rangle|0\rangle,\\[1ex]

|\psi_{13}\rangle=|1\rangle|1 - 2\rangle,~~~~~~~~~~~~
|\psi_{14}\rangle=|3\rangle|2 - 3\rangle,\\[1ex]
|\psi_{15}\rangle=|1 - 2\rangle|3\rangle,~~~~~~~~~~~~
|\psi_{16}\rangle=|2 - 3\rangle|1\rangle,\\[1ex]

|S\rangle=|0+1+2+3+4\rangle|0+1+2+3+4\rangle.
\end{array}
\end{equation}
\begin{figure}[h!]
	\includegraphics{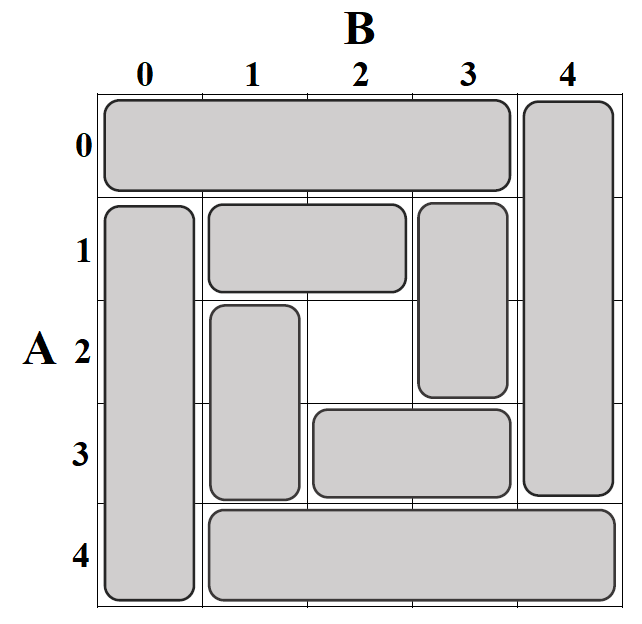}
	\caption{Tile structure of a UPB in $\mathbb{C}^5\otimes\mathbb{C}^5$.}\label{5x5}
\end{figure}

Note that the cardinality of a UPB depends on the choice of stopper state. We say the product states that are not orthogonal to the stopper $|S\rangle$ but orthogonal to all the other states in Eq.(\ref{UPB-5}), as {\it missing states}. The missing states $\{|\psi_i\rangle\}_{i=17}^{25}$ are given by-

\begin{equation}\label{rest-5}
\begin{array}{l}
|\psi_{17}\rangle=|0\rangle|0 + 1 + 2 + 3\rangle,~~
|\psi_{18}\rangle=|4\rangle|1 + 2 + 3 + 4\rangle,\\[1ex]
|\psi_{19}\rangle=|0 + 1 + 2 + 3\rangle|4\rangle,~~
|\psi_{20}\rangle=|1 + 2 + 3 + 4\rangle|0\rangle,\\[1ex]
|\psi_{21}\rangle=|1\rangle|1 + 2\rangle,~~
|\psi_{22}\rangle=|3\rangle|2 + 3\rangle,\\[1ex]

|\psi_{23}\rangle=|1 + 2\rangle|3\rangle,~~
|\psi_{24}\rangle=|2 + 3\rangle|1\rangle,~~
|\psi_{25}\rangle=|2\rangle|2\rangle.
\end{array}
\end{equation}

Note that in case of the COPB, each tile of the outer most layer corresponds to four pairwise orthogonal product states while each tile of the inner layer corresponds to two pairwise orthogonal product states the middle one corresponds to the state $|\psi_{25}\rangle=|2\rangle|2\rangle$. Because of the non-orthogonality with the stopper, one have to remove a pure product state from each tile in order to build the UPB. Hence, there are 17 states in the present UPB. This construction is new and different from that of \cite{divin2}.  

Clearly, the states $\{|\psi_i\rangle\}_{i=1}^{16}$ of Eq.(\ref{UPB-5}) and the states $\{|\psi_i\rangle\}_{i=17}^{25}$ of Eq.(\ref{rest-5}) together form a COPB in $\mathbb{C}^{5}\otimes\mathbb{C}^{5}$. Such a class of COPB in $\mathbb{C}^{d}\otimes\mathbb{C}^{d}$ and their local indistinguishability is discussed in \cite{halder}. Notice the structure given in Fig.(\ref{5x5}). To include an orthogonal product state to any of the tiles,  the new state must be orthogonal to the existing states of that tile and the stopper $|S\rangle$. But it is not possible and hence guarantees the unextendibility. Next, we construct the entangled basis $\{|\phi_i\rangle\}_{i=1}^8$ that spans the entangled subspace $\mathcal{H}_E$ of $\mathbb{C}^5\otimes\mathbb{C}^5$:

\begin{equation}\label{ent5}
\begin{array}{l}
|\phi_{1}\rangle=|\psi_{17}\rangle+|\psi_{18}\rangle-|\psi_{19}\rangle-|\psi_{20}\rangle,\\[1ex]
|\phi_{2}\rangle=|\psi_{17}\rangle-|\psi_{18}\rangle+|\psi_{19}\rangle-|\psi_{20}\rangle,\\[1ex]

|\phi_{3}\rangle=|\psi_{17}\rangle-|\psi_{18}\rangle-|\psi_{19}\rangle+|\psi_{20}\rangle,\\[1ex]

|\phi_{4}\rangle=a_3(|\psi_{17}\rangle+|\psi_{18}\rangle+|\psi_{19}\rangle+|\psi_{20}\rangle)\\[1ex]
~+a_4(a_2(|\psi_{21}\rangle+|\psi_{22}\rangle+|\psi_{23}\rangle+|\psi_{24}\rangle)-a_1|\psi_{25}\rangle),\\[1ex]

|\phi_{5}\rangle=|\psi_{21}\rangle+|\psi_{22}\rangle-|\psi_{23}\rangle-|\psi_{24}\rangle,\\[1ex]
|\phi_{6}\rangle=|\psi_{21}\rangle-|\psi_{22}\rangle+|\psi_{23}\rangle-|\psi_{24}\rangle,\\[1ex]

|\phi_{7}\rangle=|\psi_{21}\rangle-|\psi_{22}\rangle-|\psi_{23}\rangle+|\psi_{24}\rangle,\\[1ex]
|\phi_{8}\rangle=a_1(|\psi_{21}\rangle+|\psi_{22}\rangle+|\psi_{23}\rangle+|\psi_{24}\rangle)+a_2|\psi_{25}\rangle.
\end{array}
\end{equation}

The states $\{|\phi_i\rangle\}_{i=1}^8$ are pairwise orthogonal by construction. The pairwise orthogonality also holds when we consider $\{|\psi_i\rangle\}_{i=1}^{16}$ of Eq.(\ref{UPB-5}) and $\{|\phi_i\rangle\}_{i=1}^8$ of Eq.(\ref{ent5}) together. Except $|\phi_4\rangle$ and $|\phi_8\rangle$ the other states of Eq.(\ref{ent5}) are also orthogonal to the stopper $|S\rangle$ by construction. But to make $|\phi_4\rangle$ and $|\phi_8\rangle$ orthogonal to $|S\rangle$ the coefficients $a_i$'s are chosen judicially. The PPT-BE state corresponding to UPB of Eq.(\ref{UPB-5}) is of rank-8. The explicit form the state is given by:

\begin{eqnarray}\label{rho5-1}
\rho_5&=&\frac{1}{8}\left(\mathbb{I}_5\otimes\mathbb{I}_5-\sum_{i=1}^{16}|\tilde{\psi_{i}}\rangle\langle\tilde{\psi_{i}}|-|\tilde{S}\rangle\langle \tilde{S}|\right), \\\label{rho5-2}
&=&\frac{1}{8}\sum_{i=1}^{8}|\tilde{\phi_{i}}\rangle\langle\tilde{\phi_{i}}|.
\end{eqnarray}

From Observation \ref{obs1} it follows that $\rho_5$ is a PPT state. Furthermore, $\rho_5$ is supported in an entangled subspace complementary to a UPB-subspace. Therefore according to the Remark \ref{remark-1} it is an edge state. Note that, the PPT-BE state corresponding to the Tiles UPB of $\mathbb{C}^3\otimes\mathbb{C}^3$ is not only an edge state but also an extreme point of the set $\mathcal{P}$ of $\mathbb{C}^3\otimes\mathbb{C}^3$ \cite{chen}. In the following we show that $\rho_5$ is not an extreme point of the set $\mathcal{P}$ of $\mathbb{C}^5\otimes\mathbb{C}^5$, though it corresponds to a Tiles UPB.  

\begin{theorem}\label{theorem-1}
	The PPT-BE state $\rho_5$ allows a decomposition of the form $\rho_5=\frac{1}{2}\sigma_1+\frac{1}{2}\sigma_2$, where both $\sigma_1,\sigma_2$ are rank-4 PPT-BE extreme points of the set  $\mathcal{P}$ of $\mathbb{C}^5\otimes\mathbb{C}^5$.
\end{theorem}
\begin{proof}
Rewriting the Eq.(\ref{rho5-2}), we get $\rho_5=\frac{1}{2}(\frac{1}{4}\sum_{i=1}^{4}|\tilde{\phi_{i}}\rangle\langle\tilde{\phi_{i}}|)+\frac{1}{2}(\frac{1}{4}\sum_{i=5}^{8}|\tilde{\phi_{i}}\rangle\langle\tilde{\phi_{i}}|)$, where we take $\sigma_1\equiv\frac{1}{4}\sum_{i=1}^{4}|\tilde{\phi_{i}}\rangle\langle\tilde{\phi_{i}}|$ and $\sigma_2\equiv\frac{1}{4}\sum_{i=5}^{8}|\tilde{\phi_{i}}\rangle\langle\tilde{\phi_{i}}|$. Both $\sigma_1$ and $\sigma_2$ are entangled as their ranges are contained in $\mathcal{H}_E$. In particular, the range of $\sigma_2$ is contained in a two-qutrit subspace (of $\mathbb{C}^5\otimes\mathbb{C}^5$) spanned by $\{|1\rangle_A,|2\rangle_A,|3\rangle_A\}\otimes\{|1\rangle_B,|2\rangle_B,|3\rangle_B\}$. In fact, $\sigma_2$ can be written as: 
\begin{equation}\label{sigma-2}
\sigma_2=\frac{1}{4}\left(\mathbb{I}_3^\prime\otimes\mathbb{I}_3^\prime-\sum_{i=13}^{16}|\tilde{\psi}_i\rangle\langle\tilde{\psi}_i|-|\tilde{S}^\prime\rangle\langle\tilde{S}^\prime|\right),
\end{equation}
where $|S^\prime\rangle=|1+2+3\rangle|1+2+3\rangle$ and $\mathbb{I}_3^\prime$ is defined as 
\begin{equation}
\mathbb{I}_3^\prime:=|1\rangle\langle 1|+|2\rangle\langle 2|+|3\rangle\langle 3|.
\end{equation}

Observation \ref{obs1}, while employed to Eq.(\ref{sigma-2}), assures the positivity of $\sigma_2$ under partial transpose. Since, $\sigma_2$ turns out to be a rank-4 PPT-BE state of $\mathbb{C}^5\otimes\mathbb{C}^5$, therefore, according to Lemma \ref{lemma1}, $\sigma_2$ is an extreme point of the set $\mathcal{P}$ of $\mathbb{C}^5\otimes\mathbb{C}^5$.

To show the positivity of $\sigma_1$ under partial transposition, we rewrite it as:
\begin{eqnarray}\label{sigma-1}
\sigma_1=\frac{1}{4}\left(\mathbb{I}_5\otimes\mathbb{I}_5-\mathbb{I}_3^\prime\otimes\mathbb{I}_3^\prime-\sum_{i=1}^{12}|\tilde{\psi}_i\rangle\langle\tilde{\psi}_i|-|\tilde{S}\rangle\langle\tilde{S}|+|\tilde{S}^\prime\rangle\langle\tilde{S}^\prime|\right).
\end{eqnarray}
Similar argument as in the case of $\sigma_2$, applies to the above expression for $\sigma_1$ and assures the positivity of $\sigma_1$ under partial transposition. Then, applying the following mapping: $|1\rangle\rightarrow|0\rangle,~|2\rangle\rightarrow|1+2+3\rangle,~|3\rangle\rightarrow|4\rangle$ (for both Alice and Bob) to Eq.(\ref{sigma-2}), we get $\sigma_2\rightarrow\sigma_1$. This reveals the fact that the range of $\sigma_1$ is contained in a two-qutrit subspace (of $\mathbb{C}^5\otimes\mathbb{C}^5$) spanned by $\{|0\rangle_A,|1+2+3\rangle_A,|4\rangle_A\}\otimes\{|0\rangle_B,|1+2+3\rangle_B,|4\rangle_B\}$. Therefore, according to Lemma \ref{lemma1}, $\sigma_1$ is an extreme point of the set $\mathcal{P}$ of $\mathbb{C}^5\otimes\mathbb{C}^5$. This completes the proof.
\end{proof}

Please note that the two-qutrit subspaces of $\mathbb{C}^5\otimes\mathbb{C}^5$ mentioned in the above theorem (in which the ranges of states $\sigma_1$ and $\sigma_2$ are respectively contained) are not orthogonal, they have exactly one dimensional overlap. 
As already discussed in the Section \ref{sec3x3}, the entangled subspace of the Tiles UPB in $\mathbb{C}^3\otimes\mathbb{C}^3$ contains only one PPT-BE state which is also an extreme point of $\mathcal{P}$. Interestingly, for our construction in $\mathbb{C}^5\otimes\mathbb{C}^5$, there exists three PPT-BE states $\rho_5,~\sigma_1$, and $\sigma_2$. Here, $\sigma_1$ and $\sigma_2$ are extreme points of $\mathcal{P}$ as shown in Theorem \ref{theorem-1} while $\rho_5$ is an edge state but not an extreme point. Indeed, in the following corollary we give one-parameter family of PPT-BE states that are edge states. The range of all these states are contained in the same entangled subspace $\mathcal{H}_E$ as that of $\rho_5$.

\begin{corollary}\label{coro-1}	
Consider an one-parameter family of states of the form $\sigma(p):=p\sigma_1+(1-p)\sigma_2$, where $p\in[0,1]$ and $\sigma_1,\sigma_2$ are same as those used in the proof of Theorem \ref{theorem-1}. All of these states are PPT-BE edge states.
\end{corollary}

\begin{proof}
Any convex mixture of PPT states is again a PPT state, either a PPT-BE state or a separable state. As $\mathcal{R}(\sigma(p))$ is contained in fully entangled subspace $\mathcal{H}_E$, $\sigma(p)$ must be a PPT-BE state. Moreover, the Remark \ref{remark-1} guarantees it to be an edge state.
\end{proof}

Another important point is that if the Tiles UPB of  $\mathbb{C}^3\otimes\mathbb{C}^3$ is trivially extended to a higher dimensional Hilbert space (by adding suitable product states) then it is always possible to get a new UPB \cite{Roychowdhury}. Such a UPB contains $D-4$ product states, where $D$ is the net dimension of the extended system. Thus, if someone wants to construct a UPB in $\mathbb{C}^5\otimes\mathbb{C}^5$ preserving the tile structure of $\mathbb{C}^3\otimes\mathbb{C}^3$, it is then always possible to get a new UPB that contains 21 pure product states. This extension is trivial in the sense that the PPT-BE state corresponding to the higher dimensional UPB is of the same rank as old one. We, therefore, can say that our construction of Tiles UPB in $\mathbb{C}^5\otimes\mathbb{C}^5$ is a `nontrivial extension' of the Tiles UPB in $\mathbb{C}^3\otimes\mathbb{C}^3$. In the following lemma, it is shown that a nontrivial extension imposes a constraint on the cardinality of such a new UPB.    

\begin{lemma}
Preserving the tile structure of the UPB in $\mathbb{C}^3\otimes\mathbb{C}^3$ if One construct a new UPB in $\mathbb{C}^5\otimes\mathbb{C}^5$ then it is not possible to get a UPB with $n$ pure product states; where $n =18,~19,~20$. 	
\end{lemma}
\begin{proof}
We suppose that it is possible to construct a UPB in $\mathbb{C}^5\otimes\mathbb{C}^5$ with 18 pure product states, i.e., $n=18$. We also assume that the tile structure of the UPB in $\mathbb{C}^3\otimes\mathbb{C}^3$ is preserved. These result a PPT-BE state $\rho_5^\prime$ of rank-7 which can be written as $\rho_5^\prime$ = $\frac{1}{2}\sigma_1^\prime+\frac{1}{2}\sigma_2$, where $\sigma_1^\prime$ is a PPT-BE state of rank-3 (see Theorem \ref{theorem-1}). But rank-3 PPT-BE states do not exist \cite{horo3,chen-1}. Similar argument holds for $n=19,~20$. 
\end{proof}

\subsection{Tiles UPB in $\mathbb{C}^{d}\otimes\mathbb{C}^{d}$}\label{secdxd}
\begin{table}[b]
\centering
\caption{}\label{tab-1}
\begin{tabular}{c|c}
\hline\hline
Value of & States ($\omega=e^{\frac{\pi i}{k}},~i=\sqrt{-1}$,\\$k$ & $k^\prime=1,\cdots,2k-1$)\\
\hline\hline
\multirow{3}{*}{$(d-1)/2$} 
& $|0\rangle|\sum\limits_{j=0}^{d-2}\omega^{jk^\prime}(j)\rangle$,~$|d-1\rangle|\sum\limits_{j=1}^{d-1}\omega^{jk^\prime}(j)\rangle$,\\
& $|\sum\limits_{j=0}^{d-2}\omega^{jk^\prime}(j)\rangle|d-1\rangle$,~$|\sum\limits_{j=1}^{d-1}\omega^{jk^\prime}(j)\rangle|0\rangle$.\\ 
\hline
\multirow{2}{*}{$(d-3)/2$ } 
& $|1\rangle|\sum\limits_{j=1}^{d-3}\omega^{jk^\prime}(j)\rangle$,~$|d-3\rangle|\sum\limits_{j=2}^{d-2}\omega^{jk^\prime}(j)\rangle$,\\
& $|\sum\limits_{j=1}^{d-3}\omega^{jk^\prime}(j)\rangle|d-3\rangle$,~$|\sum\limits_{j=2}^{d-3}\omega^{jk^\prime}(j)\rangle|1\rangle$.\\
\hline
\multicolumn{1}{c|}{$\vdots$}  & $\vdots$\\
\hline
\multirow{2}{*}{1} 
& $|(d-3)/2\rangle|(d-3)/2-(d-1)/2\rangle$, \\
& $|(d+1)/2\rangle|(d-1)/2-(d+1)/2\rangle$, \\
& $|(d-3)/2-(d-1)/2\rangle|(d+1)/2\rangle$, \\
& $|(d-1)/2-(d+1)/2\rangle|(d-3)/2\rangle$. \\ 
\hline\hline
\end{tabular}
\end{table}

We now generalize the results of the previous section for a system in $\mathbb{C}^{d}\otimes\mathbb{C}^{d}$ with $d$ being odd. For this purpose, we first describe the tile structure given in Fig.\ref{dxd}. In this figure, there are two types of layers: type-I and type-II. Type-I corresponds to the central layer which contains only one tile. We label this central tile by $k=0$. On the other hand, all other layers are type-II layers each of which contains four tiles. We label these type-II tiles by $k=1,...,(d-1)/2$, where the outer most layer is labeled by $k=(d-1)/2$ and the inner most layer is labeled by $k=1$. Now to construct the tiles UPB, we accumulate $(2k-1)$ number of pairwise orthogonal pure product states from each of the four tiles of the $k^{th}$ layer belonging to type-II (for $k=1,2,...,(d-1)/2$). Clearly, from the $k^{th}$ layer, we take $4(2k-1)$ states. In this way, we accumulate $(d-1)^2$ states in total. Along with these states, we add a stopper $|S\rangle=|0+1+\cdots+(d-1)\rangle|0+1+\cdots+(d-1)\rangle$. These result a UPB of cardinality $(d-1)^2+1$. In Table \ref{tab-1}, we give explicit forms of the states taken from each layer (type-II) in order to construct the UPB. We label these states by $\{|\psi_i\rangle\}_{i=1}^{(d-1)^2}$.

\begin{figure}[t]
\includegraphics{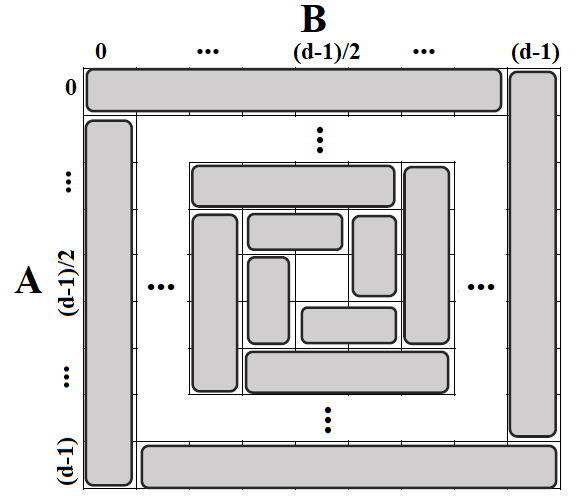}
\caption{Tile structure of a UPB in $\mathbb{C}^d\otimes\mathbb{C}^d$.}\label{dxd}
\end{figure}

\begin{table}[b!]
\caption{}\label{table-1}
\begin{tabular}{c|c}
\hline\hline
$k$ & Missing states ($t=(d-1)^2$)\\ 
\hline\hline
\multirow{4}{*}{$\frac{(d-1)}{2}$} 
& $|\psi_{t+1}\rangle=|0\rangle|0+1+\cdots+(d-2)\rangle$,~~~~~~~~~~~~~~\\
& $|\psi_{t+2}\rangle=|d-1\rangle|1+2+\cdots+(d-1)\rangle$,~~~~~~~~~\\
& $|\psi_{t+3}\rangle=|0+1+\cdots+(d-2)\rangle|d-1\rangle$,~~~~~~~~~\\
& $|\psi_{t+4}\rangle=|1+2+\cdots+(d-1)\rangle|0\rangle$,~~~~~~~~~~~~~~\\ 
\hline
\multirow{4}{*}{$\frac{(d-3)}{2}$} 
& $|\psi_{t+5}\rangle=|1\rangle|1+2+\cdots+(d-3)\rangle$,~~~~~~~~~~~~~~\\
& $|\psi_{t+6}\rangle=|d-2\rangle|2+3+\cdots+(d-2)\rangle$,~~~~~~~~~\\
& $|\psi_{t+7}\rangle=|1+2+\cdots+(d-3)\rangle|d-2\rangle$,~~~~~~~~~\\
& $|\psi_{t+8}\rangle=|2+3+\cdots+(d-2)\rangle|1\rangle$,~~~~~~~~~~~~~~\\ 
\hline
\multicolumn{1}{c|}{$\vdots$}  & $\vdots$\\
\hline
\multirow{4}{*}{1} 
& $|\psi_{d^2-4}\rangle=|(d-3)/2\rangle|(d-3)/2+(d-1)/2\rangle$, \\
& $|\psi_{d^2-3}\rangle=|(d+1)/2\rangle|(d-1)/2+(d+1)/2\rangle$, \\
& $|\psi_{d^2-2}\rangle=|(d-3)/2+(d-1)/2\rangle|(d+1)/2\rangle$, \\
& $|\psi_{d^2-1}\rangle=|(d-1)/2+(d+1)/2\rangle|(d-3)/2\rangle$, \\ \hline
\multirow{1}{*}{0} 
& $|\psi_{d^2}\rangle=|(d-1)/2\rangle|(d-1)/2\rangle$.~~~~~~~~~~~~~~~~~~~\\
\hline\hline
\end{tabular}
\end{table}

The stopper $|S\rangle$ blocks exactly one state from each of the four tiles of the $k^{th}$ layer of type-II for $k=1,2,...,(d-1)/2$, and also blocks the state of the central layer. So, to form a COPB, these $(2d-1)$ missing states along with the aforementioned $(d-1)^2$ orthogonal states are required. Here, the entangled subspace $\mathcal{H}_E$ is of dimension $2(d-1)$. Obviously, the density matrix $\rho_d$, proportional to the full rank projector onto $\mathcal{H}_E$ is of rank-$2(d-1)$. However, to construct a set of orthogonal entangled states that spans $\mathcal{H}_E$, we use the missing states as shown earlier. The missing states are listed in TABLE \ref{table-1}.

Next, we construct the set of entangled states that spans $\mathcal{H}_E$. These states are given in the TABLE \ref{table-2}. For the last states of each row, the coefficients are chosen in a way to make the state orthogonal to the stopper. Observe that the layer with $k=1$ and that with $k=0$ represent five missing states. Using these five missing states we construct four orthogonal entangled states given in the first row of TABLE \ref{table-2}. Taking an equal mixture of these four states we construct a mixed state $\sigma_{(d-1)/2}$. Clearly, $\sigma_{(d-1)/2}$ is an entangled state as its range is contained in $\mathcal{H}_E$. Moreover, following the same procedure as in case of Theorem \ref{theorem-1} it can be shown that $\sigma_{(d-1)/2}$ is indeed a PPT-BE state.

\begin{table}[t!]
\centering
\caption{}\label{table-2}
\begin{tabular}{l}
\hline\hline
~~~~~~~~~~~~~~~~~~~~Entangled states ($l=2d-2$)\\                           
\hline\hline
$|\phi_{l-3}\rangle=|\psi_{d^2-4}\rangle+|\psi_{d^2-3}\rangle-|\psi_{d^2-2}\rangle-|\psi_{d^2-1}\rangle$, \\
$|\phi_{l-2}\rangle=|\psi_{d^2-4}\rangle-|\psi_{d^2-3}\rangle+|\psi_{d^2-2}\rangle-|\psi_{d^2-1}\rangle$, \\
$|\phi_{l-1}\rangle=|\psi_{d^2-4}\rangle-|\psi_{d^2-3}\rangle-|\psi_{d^2-2}\rangle+|\psi_{d^2-1}\rangle$, \\
$|\phi_{l}\rangle=a_1(|\psi_{d^2-4}\rangle+|\psi_{d^2-3}\rangle+|\psi_{d^2-2}\rangle+|\psi_{d^2-1}\rangle)+a_2|\psi_{d^2}\rangle$, \\ 
\hline
$|\phi_{l-7}\rangle=|\psi_{d^2-8}\rangle+|\psi_{d^2-7}\rangle-|\psi_{d^2-6}\rangle-|\psi_{d^2-5}\rangle$, \\
$|\phi_{l-6}\rangle=|\psi_{d^2-8}\rangle-|\psi_{d^2-7}\rangle+|\psi_{d^2-6}\rangle-|\psi_{d^2-5}\rangle$, \\
$|\phi_{l-5}\rangle=|\psi_{d^2-8}\rangle-|\psi_{d^2-7}\rangle-|\psi_{d^2-6}\rangle+|\psi_{d^2-5}\rangle$, \\
$|\phi_{l-4}\rangle=a_3(|\psi_{d^2-8}\rangle+|\psi_{d^2-7}\rangle+|\psi_{d^2-6}\rangle+|\psi_{d^2-5}\rangle)+a_4|\phi_{l}^\perp\rangle$, \\ 
\hline
~~~~~~~~~~~~~~~~~~~~~~~~~~~~~~~~~~~~~~~~~~$\vdots$\\
\hline
$|\phi_{1}\rangle=|\psi_{t+1}\rangle+|\psi_{t+2}\rangle-|\psi_{t+3}\rangle-|\psi_{t+4}\rangle$, \\
$|\phi_{2}\rangle=|\psi_{t+1}\rangle-|\psi_{t+2}\rangle+|\psi_{t+3}\rangle-|\psi_{t+4}\rangle$, \\
$|\phi_{3}\rangle=|\psi_{t+1}\rangle-|\psi_{t+2}\rangle-|\psi_{t+3}\rangle+|\psi_{t+4}\rangle$, \\
$|\phi_{4}\rangle=a_{d-2}(|\psi_{t+1}\rangle+|\psi_{t+2}\rangle+|\psi_{t+3}\rangle+|\psi_{t+4}\rangle)+a_{d-1}|\phi_{8}^\perp\rangle$, \\
\hline\hline
\end{tabular}
\end{table}

Thereafter, we consider the missing states of the layer with $k=2$. Again, using these states we construct three entangled states given in second row (first three states) of TABLE \ref{table-2}. The last state of the same row is formed by taking linear combination of two states: first state is chosen making orthogonal to the other three states of this row and the second state is chosen making orthogonal to the last state of the previous row. The coefficients $a_3$ and $a_4$ (in second row) are chosen in a way that the state becomes orthogonal to the stopper $|S\rangle$. Using these entangled states we again construct another PPT-BT state $\sigma_{(d-3)/2}$. In this way this process is repeated up to the outermost layer ($k=(d-1)/2$) and up to PPT-BE state $\sigma_1$.

We are now ready to present Theorem \ref{theorem-2} which is a generalized version of Theorem \ref{theorem-1}. Proof of Theorem \ref{theorem-2} is straightforward from the above discussion. 

\begin{theorem}\label{theorem-2}
The PPT-BE state $\rho_d$ allows a decomposition of the form $\rho_d=\frac{2}{(d-1)}\sum_{k=1}^{(d-1)/2}\sigma_k$, where $\sigma_k$'s are rank-4 PPT-BE extreme points of the set  $\mathcal{P}$ of $\mathbb{C}^d\otimes\mathbb{C}^d$.
\end{theorem}

Note that the range of any $\sigma_k$ is contained in a two-qutrit subspaces of $\mathbb{C}^d\otimes\mathbb{C}^d$. For example, the two-qutrit subspace corresponding to $\sigma_{(d-1)/2}$ is spanned by $\{|(d-3)/2\rangle_A,|(d-1)/2\rangle_A,|(d+1)/2\rangle_A\}\otimes\{|(d-3)/2\rangle_B,|(d-1)/2\rangle_B,|(d+1)/2\rangle_B\}$; the two-qutrit subspace corresponding to $\sigma_{(d-3)/2}$ is spanned by $\{|(d-5)/2\rangle_A,|(d-3)/2+(d-1)/2+(d+1)/2\rangle_A,|(d+3)/2\rangle_A\}\otimes\{|(d-5)/2\rangle_B,|(d-3)/2+(d-1)/2+(d+1)/2\rangle_B,|(d+3)/2\rangle_B\}$, and so on. A corollary to the above theorem is stated as the following.

\begin{corollary}\label{coro-2}	
Consider a $(d-3)/2$-parameter family of states of the form $\sigma(\vec{p}):=\sum_{k=1}^{(d-1)/2}p_k\sigma_k$, where $\vec{p}$ is a probability vector of dimension $(d-1)/2$ and $\sigma_k$'s are same as in Theorem \ref{theorem-2}. All these states are PPT-BE edge states.
\end{corollary}

In Ref.\cite{leinaas2}, the authors conjectured that the rank of an extremal PPT-BE state in $\mathbb{C}^d\otimes\mathbb{C}^d$ with {\it full local ranks}, is always greater than or equal to $2(d-1)$. In the above corollary, if all $p_k$'s are nonzero then the edge states of rank $2(d-1)$ have full local ranks. Furthermore, it is an open problem whether $2(d-1)$ is also the minimum rank for any edge state to posses the property that it has full local ranks.  

\subsection{State discrimination by LOCC}
Study of UPB results in another interesting aspect called `quantum nonlocality without entanglement', where a set of bipartite product states allows local preparation but the set can non be perfectly distinguished by LOCC \cite{ben99,ben4,divin2} \footnote{Very recently, a nontrivial multipartite generalization of this phenomena is reported, where a set of multipartite product states allow local preparation but for perfect discrimination all the parties must come together or entangled resource across every bipartite cut is required \cite{mban1,mban2,mban3}.}. As already mentioned, a trivially extended UPB from $\mathbb{C}^3\otimes\mathbb{C}^3$ to $\mathbb{C}^5\otimes\mathbb{C}^5$ contains 21 pairwise orthogonal pure product states \cite{Roychowdhury}. But, our nontrivial construction in $\mathbb{C}^5\otimes\mathbb{C}^5$ contains 17 pairwise orthogonal pure product states. Both the trivial and the nontrivial constructions exhibit the phenomenon called {\it quantum nonlocality without entanglement}. Thus, it is impossible to perfectly distinguish all the states within a UPB by LOCC only. However, our nontrivial construction attributes a notable property compared to the trivial one. In the case of trivial extension it is always possible to distinguish few states perfectly from the UPB by orthogonality preserving LOCC. But in our case, not even a single state can be prefectly distinguished by such LOCC. This clearly indicates a {\it stronger notion} of local indistinguishability. Such a notion is also captured by all the higher dimensional UPBs constructed in this paper. Note that analogous notion of local indistinguishability has also been studied for multipartite systems \cite{halder1}. Now, one may raise the question that to exhibit this notion whether it is necessary to consider all the 17 states of the present UPB in $\mathbb{C}^5\otimes\mathbb{C}^5$. Interestingly, the answer is negative as it is possible to choose 9 states among these 17 states that can exhibit the aforesaid stronger notion. One possible choice of such 9 states are \{$|\psi_{2}\rangle$, $|\psi_{5}\rangle$, $|\psi_{8}\rangle$, $|\psi_{11}\rangle$, $|\psi_{13}\rangle$,$|\psi_{14}\rangle$, $|\psi_{15}\rangle$, $|\psi_{16}\rangle$, $|S\rangle$\} given in Eq.(\ref{UPB-5}). In order to distinguish this set by LOCC, at least one party has to start with a nontrivial and orthogonality preserving measurement. However, it can be shown that such a measurement does not exist for the present set. Proof follows from the argument given in \cite{Groisman, wal02}. Notice that from each tile (except the central tile) of Fig.(\ref{5x5}), we take only one state and then we add the stopper and this results a UCPB with cardinality 9 in $\mathbb{C}^5\otimes\mathbb{C}^5$. Another interesting construction is the following: It is possible to construct a set of 14 states that also exhibit the above mentioned notion but this set can be extended to a COPB. One such set of 14 states can be formed by excluding the states $|S\rangle$, $|\psi_{6}\rangle$, $|\psi_{12}\rangle$ from the UPB of Eq.(\ref{UPB-5}). A similar construction of completable set of locally indistinguishable states also follows from the Refs. \cite{Zhang14,Wang15}. However, the cardinality there will be higher than $14$. Both the UCBP and the set that is extendible to a COPB can be realized in $\mathbb{C}^d\otimes\mathbb{C}^d$.      

Our construction also leads to an interesting observation in the context of orthogonal mixed state discrimination. It is known that any state supported in the entangled subspace $\mathcal{H}_E$ cannot be distinguished unambiguously from the normalized projector onto the UPB subspace \cite{ban1}. In $\mathbb{C}^3\otimes\mathbb{C}^3$ there exists only one such which is again PPT \cite{chen}. However, our construction assures that there exists more than one PPT-BE states in $\mathbb{C}^d\otimes\mathbb{C}^d$ (with $d\ge5$ and $d$ is odd), that cannot be distinguished unambiguously from the normalized projector onto the UPB subspace. In particular, the states $\sigma(\vec{p})$ defined in Corollary \ref{coro-2} posses the above feature.

\section{Discussions}\label{discussion}
Our construction of PPT-BE states is based on Tiles UPBs. Theses UPBs can be thought of a generalization of Tiles UPB in $\mathbb{C}^3\otimes\mathbb{C}^3$. However, for higher dimension the structure is more intriguing than the simplest case. For example, the entangled subspace in $\mathbb{C}^3\otimes\mathbb{C}^3$ contains only one PPT-BE state -- the equal mixture of any orthonormal basis of that entangled subspace -- in other words the state proportional to the projector onto the entangled subspace is that PPT-BE state. All other states corresponding to unequal mixtures of the projectors of pure states forming an orthonormal basis of the entangled subspace are NPT \cite{chen}, indeed they are distillable \cite{chen1}. But our construction shows that for a higher dimensional system, the entangled subspace contains a parametric family of PPT-BE states. Here, one may raise the following question: Consider any set of orthonormal states spanning the entangled subspace for a higher dimensional system. Other than the given parametric family of states consider an arbitrary state which is a mixture of those entangled states. Will that mixed state be an NPT state? To get answer to this question, further analysis is required in this context. Another interesting research direction may be to generalize the present construction for even dimensional quantum systems and explore different properties of the corresponding PPT entangled states. It is also important to reduce the cardinality of the present UPBs so that it is possible to find new edge states with higher rank. Again, by reducing the cardinality it may be possible to find new PPT entangled states (other than rank-4) that are extremal points of the set $\mathcal{P}$. We have also discussed consequences of our construction in the context of quantum state discrimination by LOCC. We have introduced a stronger notion of local indistinguishability. Quantification of this stronger notion is another aspect for further research. It will also be interesting to find usefulness of these new PPT-BE states in quantum information processing tasks.

{\bf Acknowledgment:} SH acknowledges the financial support from The Institute of Mathematical Sciences, HBNI, Chennai where part of the work has been done when he was a visitor there. MB acknowledges support through an INSPIRE-faculty position at S. N. Bose National Centre for Basic Sciences, by the Department of Science and Technology, Government of India.

\end{document}